\documentclass[conference,a4paper]{IEEEtran}

\usepackage{url}
\usepackage{tocloft}
\usepackage{setspace}
\usepackage{amsmath}
\usepackage{multicol}
\usepackage{amssymb}
\usepackage[utf8]{inputenc}
\usepackage[T1]{fontenc}
\usepackage[font={small}]{caption}
\usepackage{cite}
\usepackage{enumerate}
\usepackage[square, comma, numbers, sort&compress]{natbib}
\usepackage{subfig}
\usepackage{bm}
\usepackage{multirow}
\usepackage{tabularx}
\usepackage{subfig}

\usepackage{url}
\usepackage{tocloft}
\usepackage{setspace}
\usepackage{amsmath}
\usepackage{multicol}
\usepackage{amssymb}
\usepackage[utf8]{inputenc}
\usepackage[T1]{fontenc}
\usepackage[font={small}]{caption}
\usepackage{cite}
\usepackage{enumerate}
\usepackage[square, comma, numbers, sort&compress]{natbib}
\usepackage{subfig}
\usepackage{bm}
\usepackage{multirow}

\newtheorem{lemma}{Lemma}
\newtheorem{proposition}{Proposition}

\makeatother

\providecommand{\theoremname}{Theorem}

\makeatother

\providecommand{\lemmaname}{Lemma}

\makeatother

\providecommand{\propositionname}{Proposition}

\makeatother

\providecommand{\corname}{Corollary}

\makeatother

\providecommand{\remname}{Remark}

\ifCLASSINFOpdf
\else
\fi
\ifCLASSINFOpdf
   \usepackage[pdftex]{graphicx}
   \usepackage{epstopdf}
\else
\fi
\begin{document}
	
	\pagenumbering{gobble}
	
	\title{A Tight Channel-Capacity Lower Bound for the Simultaneous Wireless Information and Power Transfer Integrated Receiver}

\author{
	
	\IEEEauthorblockN{Konstantinos Ntontin and Symeon Chatzinotas}
	\IEEEauthorblockA{Interdisciplinary Centre for Security, Reliability and Trust (SnT), University of Luxembourg, Luxembourg.}
	\IEEEauthorblockA{e-mails:\{kostantinos.ntontin, symeon.chatzinotas\}@uni.lu} \vspace{-2em}
}
	
	\maketitle
	
	\begin{abstract}
	    Contrary to the vast majority of works on simultaneous wireless information and power transfer that provide information-theoretic limits for the separate receiver architecture, in this work we focus on the integrated receiver and provide a channel-capacity lower bound. Towards this, we provide a closed-form tight approximation for the probability transition matrix of the channel by leveraging the 4th-order Taylor expansion of the current-voltage characteristic curve of a Schottky diode used for rectification. Numerical results reveal that the consideration of the gamma distribution as an input distribution leads to a tight channel-capacity lower bound, in contrast to other input distributions, such as the Rayleigh and uniform ones. Furthermore, the results reveal that the consideration of the 4th order term in the Taylor expansion leads to a notably higher capacity with respect to the overly simplified 2nd order term-based model.
	\end{abstract}

\section{Introduction}
\label{Introduction}

\subsection{Background}

Over the last 15 years, there has been a large body of works that provide design guidelines for simultaneous wireless information and power transfer (SWIPT). Typical applications concern low-power Internet of Things (IoT) devices that need to simultaneously be energized and provided with information data that could for instance be used for actuation. The vast majority of works have considered a SWIPT receiver architecture that dedicates different receivers for radio frequency (RF)-to-direct current (DC) power conversion and information detection. Such an architecture in the power splitting case\footnote{A similar architecture also holds for the time-splitting case.} is depicted in Fig. 1 \cite{Fundamentals_WPT_Clerckx}. $n_{th}(t)$ is the antenna thermal noise, $G_{LNA}$ is the gain of the low-noise amplification (LNA) stage, and $\rho$ is the power splitting ratio. 

\begin{figure}[!h]
\includegraphics[width=3.6in,height=1.1in]{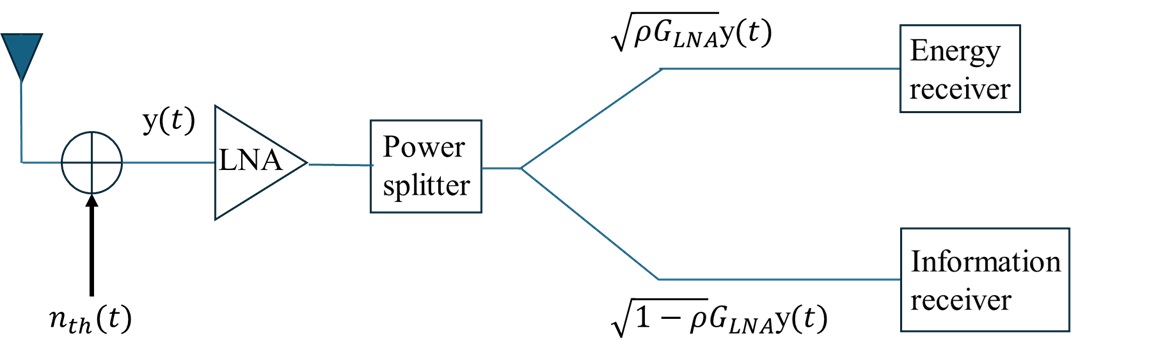}
\caption{SWIPT architecture with separate information and energy receivers.}
\end{figure}

The components of each of the 2 receivers are depicted in Fig. 2. $n_{rec}$ and $n_{mix}$ are the noise random processes introduced by the rectifier and the mixer, in the energy and information receivers, respectively.  
With the separate receiver architecture the energy and communication channel paths are separated, which means that communication channel is the well-known linear and memoryless additive white Gaussian noise (AWGN) channel. Hence, according to information theory, the rate maximization of the communication channel in the separate receiver SWIPT architecture can be achieved by a circular symmetric complex Gaussian (CSCG) input distribution. \cite{Fundamentals_WPT_Clerckx,Clerckx_Bruno_WPT_Future_Networks,Psomas_SWIPT_2024}. In fact, it has been proved that the rate-energy bound of the ideal receiver that supplies the energy and information paths without power splitting can be achieved by the power-splitting architecture of Fig. 1 \cite{Fundamentals_WPT_Clerckx}. This is possible through a superposition of a waveform that is optimal for wireless power transfer with a waveform that is optimal for information transmission \cite{Clerkx_Waveform_Design_SWIPT}.

\begin{figure}[!h]
	\centering
	\includegraphics[width=3.5in,height=2in]{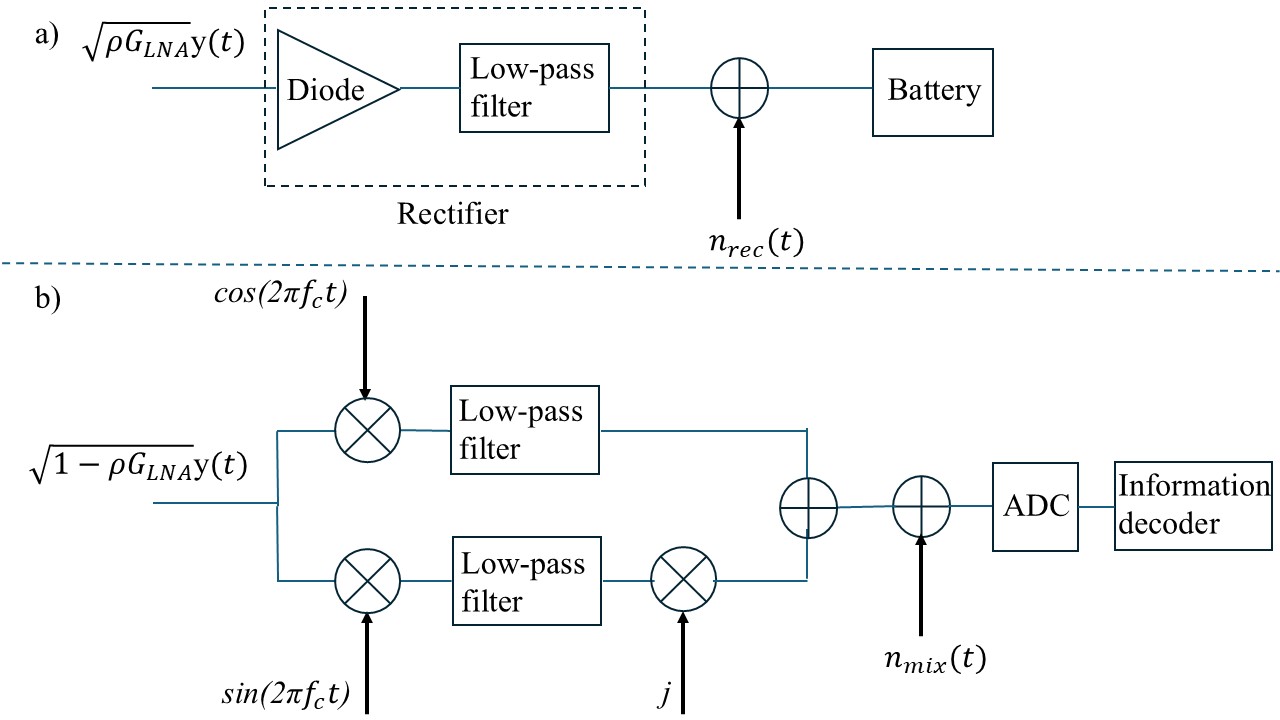}
	\caption{Components of the SWIPT separate receiver architecture: a) Energy receiver; b) Information receiver.}
\end{figure}

Despite the advantage of the separate SWIPT receiver in having disjoint communication and energy harvesting paths, its main shortcoming originates from the need for the mixing operation for downconversion that involves a local oscillator (LO). Although the technology has been notably progressing over the years and the power consumptionof LOs  has been substantially reduced, it is still in the order of hundreds of $\mu$Ws for typical sub-6 GHz 4G/5G bands \cite{Wentzloff_2020}. Such figures make it very challenging to accommodate the LO energy needs through the harvested energy from wireless power transfer. This is why several types of Ambient-IoT devices include passive envelope detectors, where only the amplitude of the signal is detected and its phase is sacrificed for achieving very low complexity \cite{Ambient_IoT_envelope_detectors}. The envelope detection inherently downconverts the RF signal to baseband and can be achieved by the same rectification circuitry that harvests energy. This is the motivation behind the SWIPT integrated receiver design, which is depicted in Fig. 3 \cite{Separate_and_Integrated_SWIPT_receiver_architectures}. $y_{DC}(t)$ is the DC signal at the output of the rectifier and ADC is the analog-to-digital converter. As it can be observed from Fig. 3, in the integrated receiver there is a common channel path for the energy harvesting and information detection processes. 

\begin{figure}[!h]
	\centering
	\includegraphics[width=3.6in,height=1in]{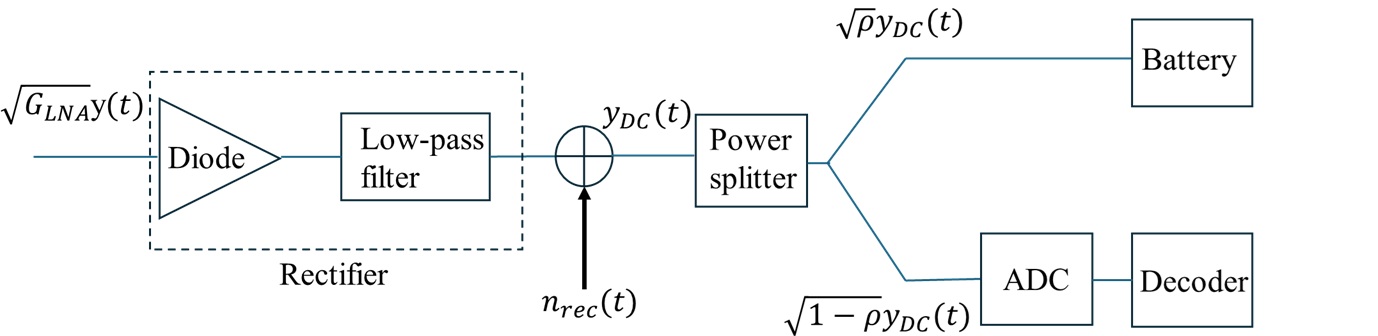}
	\caption{SWIPT integrated receiver.}
\end{figure}

The fact that in the SWIPT integrated receiver the output of the rectifier not only supplies the battery, but it is also the information signal to be decoded, creates fundamental differences with respect to the well-studied separate receiver architecture of Fig. 1. More specifically, the communication channel becomes non-linear, due to the diode used for rectification, and also subject to memory effects due to the low-pass filter. Regarding the former, a diode, normally a Schottky one, is highly non-linear circuit,  with a characteristic $i_{d}(t)-v_{d}(t)$ response that is given by

\begin{equation}
	i_d(t)=i_s(e^{\frac{v_d(t)}{nv_t}}-1).
\end{equation}

where $v_d(t)$ is the voltage drop across the diode and $i_d(t)$ is the current that flows through the diode. Moreover, $i_s$, $v_t$, $n$, are the diode reverse bias saturation current, the thermal voltage, and the ideality factor, respectively.

Through a Taylor approximation of $i_d(t)$, up to the 4th order term, the DC signal $y_{DC}$ is in the form \cite{Fundamentals_WPT_Clerckx}
\begin{align}
	y_{DC}&\approx k_2|\sqrt{{G_{LNA}}{PL}}x+\sqrt{G_{LNA}}n_{th}|^2\\\nonumber&+k_4|\sqrt{{G_{LNA}}{PL}}x+\sqrt{G_{LNA}}n_{th}|^4+n_{rec},
\end{align}
where $x$ is the transmitted signal, $PL$ is the path loss, and $k_2$ and $k_4$ are constants\footnote{The noise introduced by the ADC is not considered since its contribution is expected to be notably lower than $n_{rec}$ due to the LNA stage at the output of the receive antenna and a potentially another LNA stage before the ADC input, which is a common practice.}.

The literature is very limited with respect to works on the SWIPT integrated receiver \cite{Separate_and_Integrated_SWIPT_receiver_architectures, Tegos_2019, Kim_2022, Gkekas_2025, Demarchou_2025}. In particular, \cite{Separate_and_Integrated_SWIPT_receiver_architectures} provides the rate-energy characteristic curve, but only the 2nd-order Taylor expansion term is used for approximating $y_{DC}$. Such a rate-energy tradeoff is analytically computed also in \cite{Tegos_2019} but again only the 2nd-order Taylor expansion term is considered for $y_{DC}$.  \cite{Kim_2022} proposes a novel modulation and demodulation method for the SWIPT integrated receiver based on pulse position modulation, where information is encoded in the position of the pulse. The error rate performance and harvested DC power are numerically evaluated by considering the 4th- order Taylor expansion for $y_{DC}$. Furthermore, for a pulse position modulation the authors in \cite{Gkekas_2025} optimize the constellation design and show performance improvements with respect to a uniform design. Similar to \cite{Separate_and_Integrated_SWIPT_receiver_architectures} and \cite{Tegos_2019}, the 2nd-order Taylor expansion term is considered for $y_{DC}$. Finally, the authors in \cite{Demarchou_2025} apply circuit analysis to obtain a model that captures the memory effects of the low-pass filter. They then focus on the error rate and apply maximum likelihood sequential detection for reducing it.  

\subsection{Motivation and Contribution}
\subsubsection{Motivation}
The maximum information rate that can be transmitted over the communication channel given by (2) is unknown. The mentioned literature works that study the information-theoretic limits of the SWIPT integrated receiver consider the simplistic 2nd-order Taylor expansion for $y_{DC}$, which means that $y_{DC}\approx k_2|\sqrt{G_{LNA}PL}x+G_{LNA}n_{th}|^2+n_{rec}$. Such a second-order model normally arises in the optical domain when a photodetector is used at the receiver. Based on this, the fundamental limits of this 2nd-order response channel have been analytically studied and lower and upper bounds have been provided in the cases of $n_{th}>>n_{rec}$ \cite{Shamai_2004, Durisi_2012, Lapidoth_2002}, which is the AWGN non-coherent channel, $n_{th}<<n_{rec}$ \cite{Karout_2012, You_2002, Hranilovic_2004, Lapidoth_2009, Farid_2010}, and more recently in the general case of comparable power levels for $n_{th}$ and $n_{rec}$ \cite{Keykhosravi_2020}. Some of these bounds are tight with respect to the approximate capacity computation that originates from the implementation the Blahut-Arimoto algorithm \cite{Blahut_1972}. However, such information-theoretic frameworks have no value for the channel in (2). The corresponding channel capacity is still an open topic, as mentioned at the beginning of this subsection.

\subsubsection{Contribution}
Driven by the importance to provide bounds for the channel capacity of the channel in (2) that corresponds to the SWIPT integrated receiver, our contribution in the manuscript is the following:

\begin{itemize}
	
	\item We provide a closed-form tight approximation of the probability transition of matrix of the channel in (2).
	
	\item We provide a tight lower bound of the channel capacity, which leverages an optimized, with respect to the shape parameter, gamma distribution as the input distribution.
	
\end{itemize}

Numerical results show that the proposed lower capacity bound is tight with respect to the capacity approximation of the channel in (2), which is obtained by the Blahut-Arimoto algorithm. Furthermore, they reveal the looseness of lower bounds that leverage other distributions, such as the Rayleigh and uniform ones, especially for lower values of $G_{LNA}$. Finally, the results indicate that a notably higher channel capacity, especially for higher values of $G_{LNA}$, is achieved by the 4th-order approximation of the characteristic curve of the diode compared to the inaccurate 2nd-order approximation.

The rest of the manuscript is structured as follows:
In Section II, the system model is presented. Section III provides the analytical computation for the probability transition matrix of the channel capacity expression. In Section IV we present numerical results that originate from the analytical expressions. Finally, Section V concludes this work.

\emph{Notation throughout the manuscript}: i) $E\{\cdot\}$ denotes average value; ii)  $N(\mu,\sigma^2)$ denotes a  normal variable with mean $\mu$ and variance $\sigma^2$; iii) $p(\cdot)$ denotes a pdf. 
\section{System model}
\begin{figure}[!h]
	\centering
	\includegraphics[width=3.4in,height=2in]{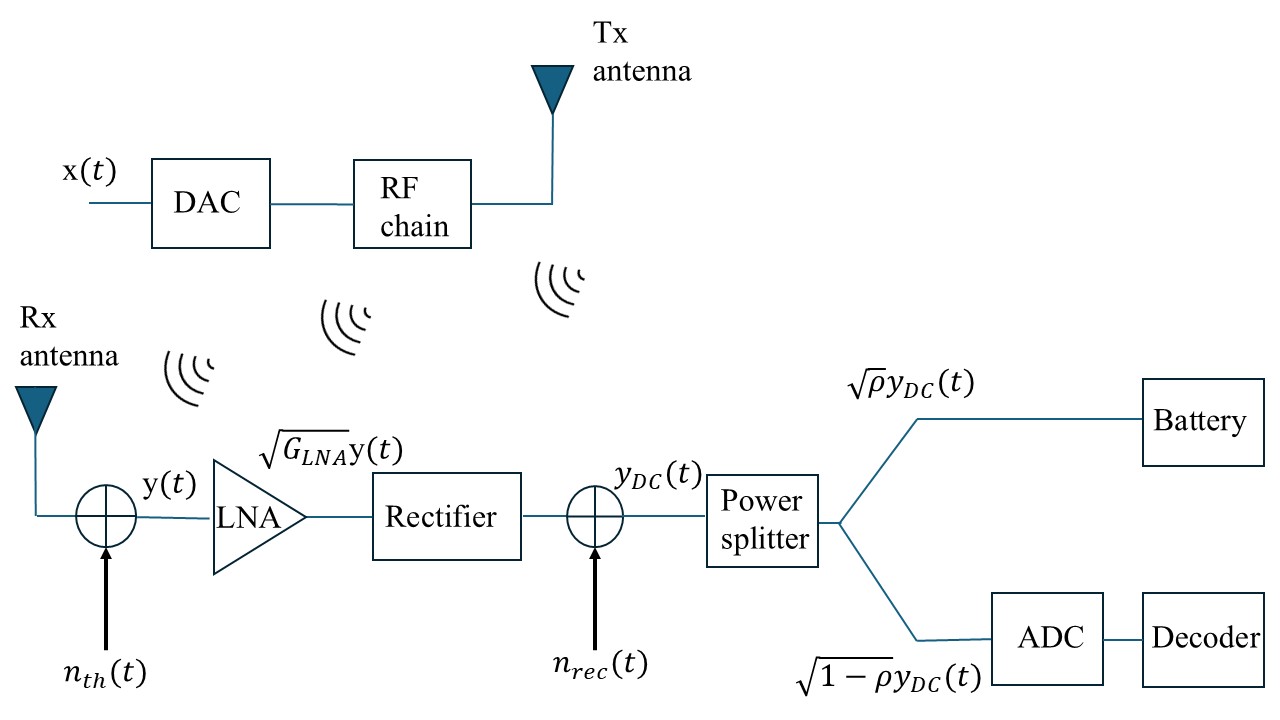}
	\caption{System model.}
\end{figure}
We consider the system model presented in Fig. 4, where $x(t)$ is the transmit signal. We assume that the propagation channel is a line-of-sight channel, which means that it follows the Friis free-space propagation law. Hence, for the received signal power, denoted by $P_r$, it holds
\begin{equation}
P_r=P_tG_tG_r(\frac{\lambda}{4\pi d})^2,
\end{equation}
where $P_t$, $G_t$, $G_r$, $\lambda$, and $d$ are the transmit power, transmit antenna gain, receive antenna gain, wavelength, and transmitter- receiver distance, respectively.
Consequently, by using the 4th-order Taylor approximation for the current-voltage characteristic curve of the diode, for $y_{DC}$ it holds
\begin{align}
	y_{DC}\approx k_2|\tilde{x}+\tilde{n}_{th}|^2+k_4|\tilde{x}+\tilde{n}_{th}|^4+n_{rec}, 
\end{align}
where $\tilde{x}=\sqrt{G_{LNA}P_r}x$, $\tilde{n}_{th}=\sqrt{G_{LNA}}n_{th}$, and 
\begin{align}
	k_i=\frac{i_sR_{ant}^{i/2}}{i!(nv_t)^i}, \:\:\:\:\:\: i=1, 2.
\end{align}
 Furthermore, $R_{ant}$ is the antenna impedance. Finally, $n_{th}\sim CN(0,P_{th})$, $n_{rec}\sim N(0,P_{rec})$, and $P_{th}=k_0TB$, where $k_0$ is Boltzmann's constant, $T$ the temperature in Kelvin, and $B$ the signal bandwidth.

\section{Channel-Capacity Analysis}

For the capacity $C$ of the channel in (4), we have
\begin{align}
	C&=\underset{p(\tilde{x})}{\text{max}}\:I(\tilde{x};y_{DC})\overset{(a)}{=}\underset{p(\tilde{x})}{\text{max}}\:I(\tilde{x}^2;y_{DC})\overset{(b)}{=}\underset{p(u)}{\text{max}}\:I(u;y_{DC})\nonumber\\
	&\text{subject to} \: \: E\{u\}= G_{LNA}P_r.
\end{align}
where $\tilde{x}$ is a continuous vaiable and $I(\cdot;\cdot)$ denotes the mutual information. (a) is due to the fact that $\tilde{x}\ge0$ and in (b) we set $u=\tilde{x}^2$. In addition, $I(u;y_{DC})$ is given by \cite{proakis2001digital}
\begin{align}
	I(u;y_{DC})\!=\!\!\int_{0}^{\infty}\!\!\int_{-\infty}^{\infty}\!p(y_{DC}|u)p(u)\log_2(\frac{p(y_{DC}|u)}{p(y_{DC})})dy_{DC}du.
\end{align} 
$p(y_{DC}|u)$ is the probability transition matrix of the channel in (4). Before providing an approximate expression for $p(y_{DC}|u)$, let us first present Lemma 1. 
\begin{lemma}
	A noncentral chi-squared distribution with $k$ degrees of freedom and non-centrality parameter $s$ is approximated for large $s$ by a normal distribution with mean $k+s$ and variance $2(k+2s)$, which are the mean and variance of noncentral chi-squared distribution. As $s \to \infty$, the noncentral chi-square distribution asymptotically tends to $N(k+s,2(k+2s))$.
	\end{lemma}
	
	\emph{Proof}: The proof directly follows by considering the skewness of the noncentral chi-squared distribution, which is equal to $\frac{2^(\frac{3}{2})(k+3s)}{(k+2s)^{\frac{3}{2}}}$ and is a measure of its deviation from a normal variable with the same mean and variance. For ${s \to \infty}$, the skewness tends to 0. This concludes the proof. 
	
	To substantiate the suitability of the normal distribution as an approximation for large s even further, for $k=2$ in Fig. 5 we illustrate $\int_{0}^{\infty}|F_{nccs}(x)-F_{normal}(x)|^2dx$ vs. s, $x\ge0$, where $F_{nccs}(x)$ is the cumulative distribution function (CDF) of the noncentral chi-squared distribution and $F_{normal}(x)$ is the CDF of a normal distribution with the same mean and variance.
	\begin{figure}[!h]
		\centering
		\includegraphics[width=3.2in,height=1.8in]{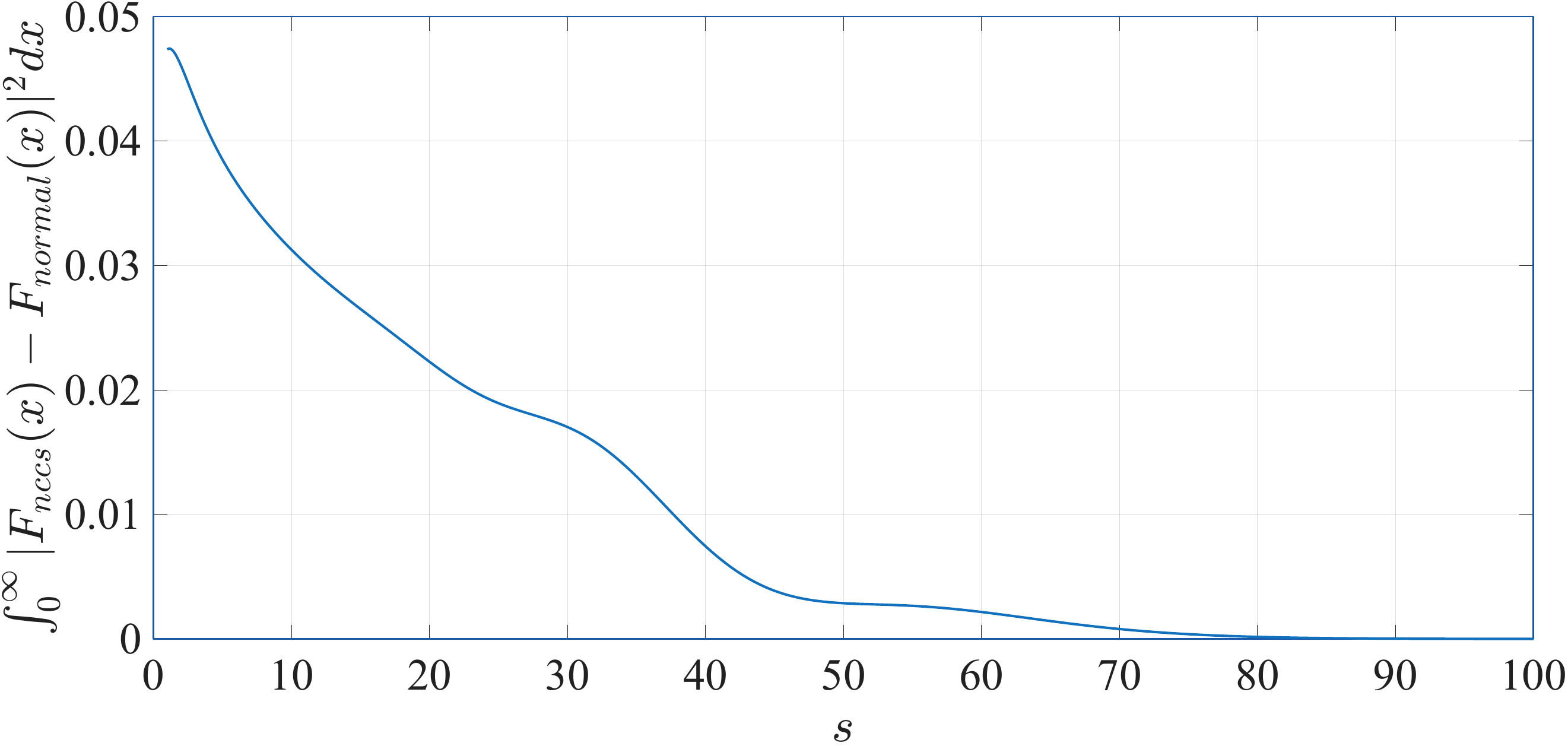}
		\caption{$\int_{0}^{\infty}|F_{nccs}(x)-F_{normal}(x)|^2dx$ vs. $s$ for $k=2$.}
	\end{figure}
		From Fig. 5 we observe that $\int_{0}^{\infty}|F_{nccs}(x)-F_{normal}(x)|^2dx$ quickly goes to 0 as $s$ increases.
	\begin{proposition}
		$p(y_{DC}|u)$ is approximated as
		\begin{align}
			p(y_{DC}|u)\cong \frac{1}{\sqrt{2\pi \sigma_{y_{DC}}^2}}\exp(-\frac{(y_{DC}-\mu_{y_{DC}})^2}{2\sigma_{y_{DC}}^2}),
		\end{align}
		where
		\begin{equation}
		\begin{gathered}
			\mu_{y_{DC}}=a_4\sigma_u^2(1+\frac{(2a_4\mu_u+a_2)^2}{4a_4\sigma_u^2})-\frac{a_2^2}{4a_4}\\
			\sigma_{y_{DC}}^2=2(a_4\sigma_u^2)^2(1+2\frac{(2a_4\mu_u+a_2)^2}{4a_4\sigma_u^2})+P_{rec}
			\end{gathered}
			\end{equation}
			with
			\begin{equation}
			\begin{gathered}
				a_2=k_2G_{LNA}P_{th}/2, \; \; \; \; a_4=k_4(G_{LNA}P_{th}/2)^2\\
				\mu_u=2+\frac{2}{G_{LNA}P_{th}}u, \; \; \sigma_u^2=2(2+2\frac{2}{G_{LNA}P_{th}}u)
				\end{gathered}
				\end{equation}
	\end{proposition}
	\emph{Proof}: See the Appendix.
	
	$p(y_{DC})$ in (7) is given by $p(y_{DC})=\int_0^{\infty}p(y_{DC}|u)p(u)du$. 
	The importance of using Lemma 1 in the proof of Proposition 1 lies on the avoidance of having the modified Bessel function term in the expression of  $p_{x|u}(x|u)$ that is associated with the pdf of a noncentral chi-squared distribution. Thus, the time needed for the numerical computation of the channel capacity in (7) can be notably reduced.
	
	For the proposed lower bound we consider the gamma distribution due to its flexibility. Indicatively, the gamma distribution for obtaining a channel-capacity lower bound is considered in \cite{Keykhosravi_2020} and the bound is shown to be tight with respect to the approximate channel capacity obtained by the Blahut-Arimoto algorithm. 
	Based on the gamma distribution assumption for $u$, it holds
	\begin{align}
		p(u)=\frac{\alpha^{\alpha}}{\Gamma(\alpha)(G_{LNA}P_r)^{\alpha}}u^{\alpha-1}\exp(-\frac{\alpha}{G_{LNA}P_r}u),
	\end{align}
	where $\alpha$ is the shape parameter of the gamma distribution.
	The channel-capacity lower bound is obtained by optimizing over $\alpha$, so as to maximize $I(u;y_{DC})$.

\section{Numerical results}

We consider the parameters of Table I.

\begin{table}[h]
	\label{Parameter_values}
	\caption{Parameter values used in the simulation.} 
	\centering 
	\scalebox{0.8}{
		\begin{tabular}{| c | c | } 
			\hline
			Parameter & Value \\[0.5ex]
			\hline
			\hline
			$\lambda$& $0.1$ m (corresponding to a 3 GHz carrier frequency) \\[0.5ex]
				\hline
			$d$& $10$ m \\[0.5ex] 
			\hline
			$P_{t}$& $1$ W  \\ [0.5ex]
			\hline
		     $G_{t}$& $20$ dBi  \\ [0.5ex]
		     \hline
		     $G_{r}$& $3$ dBi  \\ [0.5ex]
		     \hline
		     $i_s$& $5$ $\mu$A \cite{Clerckx_2016_Waveform-design_WPT}   \\ [0.5ex]
		     \hline
		     $n$& $1.05$ \cite{Clerckx_2016_Waveform-design_WPT} \\ [0.5ex]
		     \hline
		     $v_t$& $25.86$ mV \cite{Clerckx_2016_Waveform-design_WPT}  \\ [0.5ex]
		     \hline
		      $B$& $10 $ MHz  \\ [0.5ex]
		     \hline
		      $T$& $300 $ K  \\ [0.5ex]
		      \hline
			$P_{\textrm{rect}}$& $10^{3}P_{th}$  \\ [0.5ex]
			\hline
	\end{tabular}}
	\label{Parameter_values} 
\end{table}

\begin{figure}[!h]
	\centering
	\includegraphics[width=3.2in,height=2.2in]{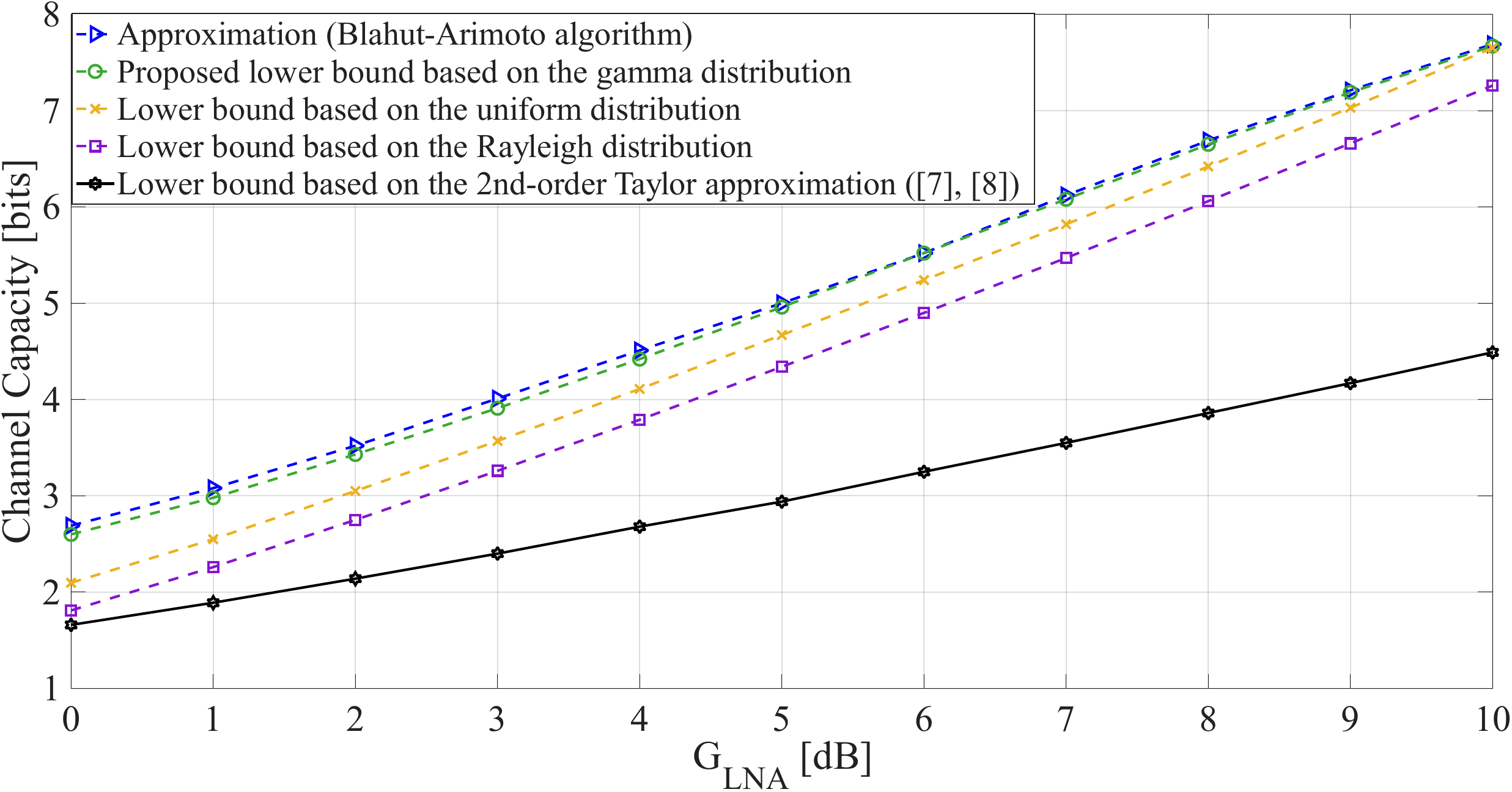}
	\caption{Channel capacity.}
\end{figure}

Fig. 6 illustrates the channel capacity vs. $G_{LNA}$ curves for  the approximate expression based on the Blahut-Arimoto algorithm, our proposed lower bound based on the gamma distribution, two lower bounds based on the uniform and Rayleigh distributions of the same mean with the gamma distribution and, finally, a lower bound that considers only the 2nd- order term of the Taylor approximation. As we observe from Fig. 6, the proposed lower bound is tight with respect to the approximate channel capacity. In particular, for low to moderate values of $G_{LNA}$ it exhibits a gap of only 0.1 bits. On the other hand, the uniform distribution-based bound exhibits a gap of around 0.5 bits for low-to-moderate $G_{LNA}$ values. The gap is notably higher for the Rayleigh-based bound. Finally, we observe from Fig. 6 that the overly simplified diode model that relies on the 2nd-order only Taylor expansion leads to a notable underestimation of the channel capacity. This is an indication of how important it is to leverage the 4th-order Taylor expansion of the diode characteristic current-voltage curve for accurate computation of the channel capacity.

\section{Conclusions}

 This work tried for the first time  to derive channel-capacity bounds for the SWIPT integrated receiver under realistic receiver hardware assumptions. Towards this, we have focused on the provision of a lower bound and leveraged an accurate 4th-order Taylor expansion representation of the current-voltage characteristic curve of the diode used for rectification, in contrast to previous works. Based on this, we have provided a closed-form tight approximation for the probability transition matrix of the channel. Furthermore, for computing the  channel-capacity lower bound, an optimized, with respect to the shape parameter, gamma distribution has been considered. 
 
 The simulation results have revealed that the proposed bound is tight with respect to the approximate channel capacity, in contrast to other considered distributions, such as the uniform and Rayleigh ones. Moreover, the results have showed that the overly simplistic lower bound that leverages only the 2nd-order Taylor expansion, and has widely been used in the literature,  results in substantially lower channel-capacity values. The gap increases for increasing $G_{LNA}$. 
 There are several avenues for future work. More specifically, we are interested to also provide a tight upper bound for the channel capacity. In addition, the impact of memory effects on the channel capacity, that arise from the low-pass filter, will also be investigated. Furthermore, the inclusion of the wireless power transfer requirements as a constraint in the mutual information maximization problem in (6) will result in a different optimal input distribution. Finally, the study of practical modulation and coding schemes to achieve the channel capacity is another topic of interest.

\section*{Appendix}

\emph{Proof of Proposition 1}: According to (4) and the independence of $\tilde{n}_{th}$ and ${n}_{rec}$, it holds that $p(y_{DC}|u)$ is given by the convolution of the pdfs of the independent random processes $k_2|\sqrt{u}+\tilde{n}_{th}|^2+k_4|\sqrt{u}+\tilde{n}_{th}|^4$ for fixed $u$, and $n_{rec}$. As far as the random process $|\sqrt{u}+\tilde{n}_{th}|^2$ is concerned, it follows a noncentral chi-squared distribution with 2 degrees of freedom and non-centrality parameter equal to $\frac{2}{G_{LNA}P_{th}}u$. By taking into account that $P_{th}<<P_r$  in practical cases, it holds
\begin{align}
	E\{\frac{2}{G_{LNA}P_{th}}u\}=2\frac{P_r}{P_{th}}>>1.
\end{align}
Hence, based on Lemma 1, $|\sqrt{u}+\tilde{n}_{th}|^2$ is approximated by a normal random process with mean equal to $\mu_u$ and variance equal to $\sigma_u^2$. By denoting $|\sqrt{u}+\tilde{n}_{th}|^2$ as $h$, the process $z=k_2h+k_4h^2$, after completing the square, is equal to
\begin{align}
	z=a_4\sigma_u^2\chi_1(s_u)-\frac{a_2^2}{4a_4},
\end{align}
where $\chi_1(s_u)$ follows a noncentral chi-squared distribution with 1 degree of freedom and noncentrality parameter $s_u$ given by
\begin{align}
	s_u=\frac{(2a_4\mu_u+a_2)^2}{4a_4\sigma_u^2}.
\end{align}
It holds that
\begin{align}
	E\{s_u\}\propto E\{u\}=G_{LNA}P_r>>1.
\end{align}
Hence, by again incorporating Lemma 1, we approximate $\chi_1(s_u)$ by a normal distribution with mean $1+s_u$ and variance $2(1+2s_u)$. As a result, $z$ is approximated by a normal variable with mean $\mu_{z_u}=a_4\sigma_u^2(1+s_u)-\frac{a_2^2}{4a_4}$ and variance $\sigma_{z_u}^2=2(a_4\sigma_u^2)^2(1+2s_u)$. Consequently, the pdf of $z$, which we denote by $p_z(z)$, is given by
\begin{align}
	p_z(z)=\frac{1}{\sqrt{2\pi \sigma_{z_u}^2}}\exp(-\frac{(z-\mu_{z_u})^2}{2\sigma_{z_u}^2}).
\end{align}
By denoting the pdf of $n_{rec}$ by $p_{rec}(x)$, it holds that
\begin{align}
		p(y_{DC}|u)&\cong (p_z*p_{rec})(y_{DC})\\\nonumber&\overset{(c)}{=}\frac{1}{\sqrt{2\pi (\sigma_{z_u}^2+P_{rec})}}\exp(-\frac{(z-\mu_{z_u})^2}{2(\sigma_{z_u}^2+P_{rec})}),
\end{align}
where * denotes convolution and (c) is due to the fact that the convolution of 2 independent normal variables is equal to a normal variable with mean the sum of their means and variance the sum of their variances. This concludes the proof.



\bibliographystyle{IEEEtran}
\footnotesize{
	\bibliography{IEEEabrv,references}

\begin{thebibliography}{10}
\providecommand{\url}[1]{#1}
\csname url@samestyle\endcsname
\providecommand{\newblock}{\relax}
\providecommand{\bibinfo}[2]{#2}
\providecommand{\BIBentrySTDinterwordspacing}{\spaceskip=0pt\relax}
\providecommand{\BIBentryALTinterwordstretchfactor}{4}
\providecommand{\BIBentryALTinterwordspacing}{\spaceskip=\fontdimen2\font plus
\BIBentryALTinterwordstretchfactor\fontdimen3\font minus
  \fontdimen4\font\relax}
\providecommand{\BIBforeignlanguage}[2]{{%
\expandafter\ifx\csname l@#1\endcsname\relax
\typeout{** WARNING: IEEEtran.bst: No hyphenation pattern has been}%
\typeout{** loaded for the language `#1'. Using the pattern for}%
\typeout{** the default language instead.}%
\else
\language=\csname l@#1\endcsname
\fi
#2}}
\providecommand{\BIBdecl}{\relax}
\BIBdecl

\bibitem{Fundamentals_WPT_Clerckx}
B.~Clerckx, R.~Zhang, R.~Schober, D.~W.~K. Ng, D.~I. Kim, and H.~V. Poor,
  ``Fundamentals of wireless information and power transfer: From rf energy
  harvester models to signal and system designs,'' \emph{IEEE Journal on
  Selected Areas in Communications}, vol.~37, no.~1, pp. 4--33, 2019.

\bibitem{Clerckx_Bruno_WPT_Future_Networks}
B.~Clerckx, K.~Huang, L.~R. Varshney, S.~Ulukus, and M.-S. Alouini, ``Wireless
  power transfer for future networks: Signal processing, machine learning,
  computing, and sensing,'' \emph{IEEE Journal of Selected Topics in Signal
  Processing}, vol.~15, no.~5, pp. 1060--1094, 2021.

\bibitem{Psomas_SWIPT_2024}
C.~Psomas, K.~Ntougias, N.~Shanin, D.~Xu, K.~Mayer, N.~M. Tran,
  L.~Cottatellucci, K.~W. Choi, D.~I. Kim, R.~Schober, and I.~Krikidis,
  ``Wireless information and energy transfer in the era of 6g communications,''
  \emph{Proceedings of the IEEE}, vol. 112, no.~7, pp. 764--804, 2024.

\bibitem{Clerkx_Waveform_Design_SWIPT}
B.~Clerckx, ``Waveform optimization for swipt with nonlinear energy harvester
  modeling,'' in \emph{WSA 2016; 20th International ITG Workshop on Smart
  Antennas}, 2016, pp. 1--5.

\bibitem{Wentzloff_2020}
D.~D. Wentzloff, A.~Alghaihab, and J.~Im, ``Ultra-low power receivers for iot
  applications: A review,'' in \emph{2020 IEEE Custom Integrated Circuits
  Conference (CICC)}, 2020, pp. 1--8.

\bibitem{Ambient_IoT_envelope_detectors}
A.~Al-nahari, J.~Liao, R.~Jäntti, D.~Mishra, D.-T. Phan-Huy, and Y.~Zhou,
  ``Ambient iot connectivity topologies: Technology enablers, applications, and
  challenges,'' \emph{IEEE Internet of Things Magazine}, pp. 1--8, 2025.

\bibitem{Separate_and_Integrated_SWIPT_receiver_architectures}
X.~Zhou, R.~Zhang, and C.~K. Ho, ``Wireless information and power transfer:
  Architecture design and rate-energy tradeoff,'' \emph{IEEE Transactions on
  Communications}, vol.~61, no.~11, pp. 4754--4767, 2013.

\bibitem{Tegos_2019}
S.~A. Tegos, P.~D. Diamantoulakis, K.~N. Pappi, P.~C. Sofotasios, S.~Muhaidat,
  and G.~K. Karagiannidis, ``Toward efficient integration of information and
  energy reception,'' \emph{IEEE Transactions on Communications}, vol.~67,
  no.~9, pp. 6572--6585, 2019.

\bibitem{Kim_2022}
J.~Kim and B.~Clerckx, ``Wireless information and power transfer for iot: Pulse
  position modulation, integrated receiver, and experimental validation,''
  \emph{IEEE Internet of Things Journal}, vol.~9, no.~14, pp. 12\,378--12\,394,
  2022.

\bibitem{Gkekas_2025}
A.~Gkekas, V.~E. Galanopoulou, O.~G. Karagiannidis, S.~A. Tegos, and P.~D.
  Diamantoulakis, ``Optimal pulse energy modulation for swipt systems with
  integrated receivers,'' \emph{IEEE Communications Letters}, vol.~29, no.~5,
  pp. 1028--1031, 2025.

\bibitem{Demarchou_2025}
E.~Demarchou, Z.~Bin~Ashraf, B.~Smida, C.~Psomas, and I.~Krikidis, ``Integrated
  swipt receivers: Circuit analysis and performance evaluation,'' \emph{IEEE
  Transactions on Communications}, vol.~73, no.~11, pp. 10\,345--10\,359, 2025.

\bibitem{Shamai_2004}
M.~Katz and S.~Shamai, ``On the capacity-achieving distribution of the
  discrete-time noncoherent and partially coherent awgn channels,'' \emph{IEEE
  Transactions on Information Theory}, vol.~50, no.~10, pp. 2257--2270, 2004.

\bibitem{Durisi_2012}
G.~Durisi, ``On the capacity of the block-memoryless phase-noise channel,''
  \emph{IEEE Communications Letters}, vol.~16, no.~8, pp. 1157--1160, 2012.

\bibitem{Lapidoth_2002}
A.~Lapidoth, ``On phase noise channels at high snr,'' in \emph{Proceedings of
  the IEEE Information Theory Workshop}, 2002, pp. 1--4.

\bibitem{Karout_2012}
J.~Karout, E.~Agrell, K.~Szczerba, and M.~Karlsson, ``Optimizing constellations
  for single-subcarrier intensity-modulated optical systems,'' \emph{IEEE
  Transactions on Information Theory}, vol.~58, no.~7, pp. 4645--4659, 2012.

\bibitem{You_2002}
R.~You and J.~Kahn, ``Upper-bounding the capacity of optical im/dd channels
  with multiple-subcarrier modulation and fixed bias using trigonometric moment
  space method,'' \emph{IEEE Transactions on Information Theory}, vol.~48,
  no.~2, pp. 514--523, 2002.

\bibitem{Hranilovic_2004}
S.~Hranilovic and F.~Kschischang, ``Capacity bounds for power- and band-limited
  optical intensity channels corrupted by gaussian noise,'' \emph{IEEE Trans.
  on Information Theory}, vol.~50, no.~5, pp. 784--795, 2004.

\bibitem{Lapidoth_2009}
A.~Lapidoth, S.~M. Moser, and M.~A. Wigger, ``On the capacity of free-space
  optical intensity channels,'' \emph{IEEE Transactions on Information Theory},
  vol.~55, no.~10, pp. 4449--4461, 2009.

\bibitem{Farid_2010}
A.~A. Farid and S.~Hranilovic, ``Capacity bounds for wireless optical intensity
  channels with gaussian noise,'' \emph{IEEE Transactions on Information
  Theory}, vol.~56, no.~12, pp. 6066--6077, 2010.

\bibitem{Keykhosravi_2020}
K.~Keykhosravi, E.~Agrell, M.~Secondini, and M.~Karlsson, ``When to use optical
  amplification in noncoherent transmission: An information-theoretic
  approach,'' \emph{IEEE Transactions on Communications}, vol.~68, no.~4, pp.
  2438--2445, 2020.

\bibitem{Blahut_1972}
R.~Blahut, ``Computation of channel capacity and rate-distortion functions,''
  \emph{IEEE Transactions on Information Theory}, vol.~18, no.~4, pp. 460--473,
  1972.

\bibitem{proakis2001digital}
J.~G. Proakis, \emph{Digital Communications}, 4th~ed.\hskip 1em plus 0.5em
  minus 0.4em\relax New York: McGraw-Hill, 2001, iSBN-13: 978-0072321111.

\bibitem{Clerckx_2016_Waveform-design_WPT}
B.~Clerckx and E.~Bayguzina, ``Waveform design for wireless power transfer,''
  \emph{IEEE Transactions on Signal Processing}, vol.~64, no.~23, pp.
  6313--6328, 2016.

\end{thebibliography}
}

\end{document}